\documentclass[a4paper]{article}

\usepackage{arydshln}
\usepackage{graphicx}
\usepackage{amsmath}
\usepackage{amssymb, verbatim}
\usepackage{mathrsfs}
\usepackage{dsfont}
\usepackage{url}
\usepackage[usenames,dvipsnames]{color}
\usepackage[compress]{cite}
\usepackage{mdwlist}
\usepackage{paralist}
\usepackage{algorithmic}
\usepackage{algorithm}
\usepackage{tikz}
\usepackage{tkz-graph}

\newtheorem{lemma}{\sc Lemma}
\newtheorem{theorem}[lemma]{\sc Theorem}

\newtheorem{corollary}[lemma]{\sc Corollary}
\newtheorem{remark}{\sc Remark}
\newtheorem{assumption}{\sc Assumption}
\newtheorem{definition}{\sc Definition}

\PassOptionsToPackage{usenames,dvipsnames,svgnames}{xcolor}
\usetikzlibrary{arrows,positioning,automata}
\usetikzlibrary{shadows}
\usetikzlibrary{calc}
\usetikzlibrary{shapes}

\renewcommand{\footnoterule}{%
  \kern -7pt
  \hrule width 0.3\textwidth height .5pt
  \kern 2pt
}

\renewcommand{\matrix}[2]{\left[\begin{array}{#1} #2 \end{array}\right] }

\def\Red#1{\textcolor{black}{#1}}
\newcommand{\blue}[1]{\color{black}#1\color{black}}

\DeclareMathOperator*{\argmin}{arg\,min}
\DeclareMathOperator*{\argmax}{arg\,max}
\DeclareMathOperator*{\spans}{span}
\DeclareMathOperator*{\supp}{support}
\DeclareMathOperator*{\modd}{mod}
\DeclareMathOperator*{\diag}{diag}

\floatname{algorithm}{Procedure}

\newcommand{\IEEEQED}{~\rule[-1pt]{5pt}{5pt}\par\medskip}
\newenvironment{proof}{{\bf Proof:\ }}{ \hfill \IEEEQED}

\begin{document}

\title{Networked Estimation using Sparsifying Basis Prediction\thanks{The work of F.~Farokhi and K.~H.~Johansson was supported by the Swedish Research Council and the Knut and Alice Wallenberg Foundation.}}

\author{Farhad Farokhi, Amirpasha~Shirazinia, and Karl H. Johansson\thanks{ACCESS Linnaeus Center, School of Electrical Engineering, KTH Royal Institute of Technology, Stockholm, Sweden. Emails:\{farakhi,amishi,kallej\}@kth.se}}

\date{}

\maketitle

\begin{abstract}
We present a framework for networked state estimation, where systems encode their (possibly high dimensional) state vectors using a mutually agreed basis between the system and the estimator (in a remote monitoring unit). The basis sparsifies the state vectors, i.e., it represents them using vectors with few non-zero components, and as a result, the systems might need to transmit only a fraction of the original information to be able to recover the non-zero components of the transformed state vector. Hence, the estimator can recover the state vector of the system from an under-determined linear set of equations. We use a greedy search algorithm to calculate the sparsifying basis. Then, we present an upper bound for the estimation error. Finally, we demonstrate the results on a numerical example.
\end{abstract}

\section{Introduction}
\blue{Networked monitoring and estimation, where sensors, estimators, and monitoring units communicate over a shared medium (e.g, a common communication bus, a wireless network, Internet, Ethernet), has attracted much attention recently because of flexible maintenance and upgrades~\cite{4341568,6263277,taylor2006networked,
Katewa5991185,epstein2008probabilistic}.
However, shared communication mediums bring some limitations, such as band-limited channels, variable delays, and packet drop-outs. As an example, consider the schematic diagram in Figure~\ref{figure:network}, where $N$ systems (denoted $P_i$) are trying to communicate their state measurements to a monitoring unit at a distant location (denoted $M$). Such systems are found in many industrial domains, e.g., monitoring applications for large chemical plants where thousands of sensors are communicating measurements to a central operator. Each system itself is composed of many subsystems, as illustrated for system $P_1$ in Figure~\ref{figure:network}, where each subsystem is denoted by $P_1^\ell$ for $1\leq \ell\leq L=6$. 
The overall state of all these systems is too large to be communicated in real-time over a conventional communication network, but we need to reduce its dimension to achieve a solution that can be implemented on low-cost hardware. In this paper, we propose an encoding/decoding algorithm for each individual system $P_i$, $1\leq i\leq N$, so that it needs to transmit an output vector with only a fraction of its state vector dimension. To do so, we utilize the idea of sparsifying bases (also known as the sparsifying dictionaries or transforms) from the compressive sensing literature.}

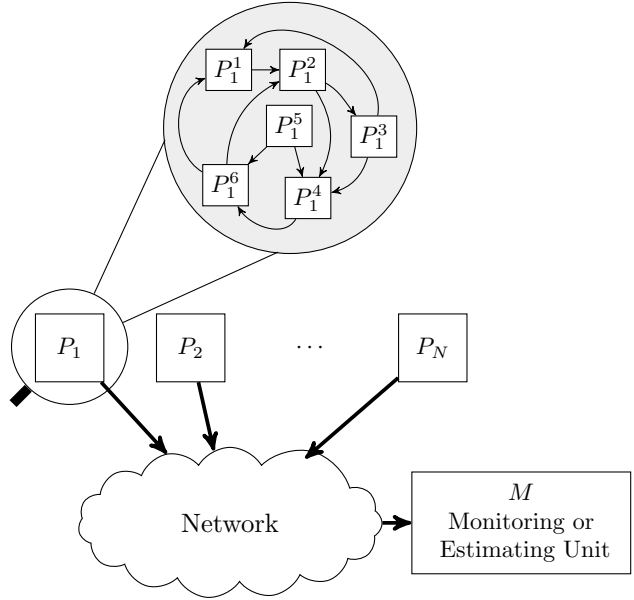
\begin{figure}[t]
\centering
\hspace{-.7in}
\begin{tikzpicture}[>=stealth',node distance=.68cm,initial/.style={}]
\node[circle,draw,minimum size=1.7cm,scale=0.9,fill opacity=0.2,node distance=1.16cm] (Z1){};
\node[rectangle,node distance=1.44cm] (I1)[below =of Z1] {};
\node[rectangle,draw,minimum size=1cm,scale=0.9,node distance=1.74cm] [above =of I1] (P1) {$P_1$};
\node[rectangle,minimum size=1cm,scale=0.9,node distance=0.4cm] [below left =of P1] (B1) {};
\path (B1) edge[-,line width=0.13cm] (Z1);
\node[rectangle,draw,fill=white,minimum size=1cm,scale=0.9] (P2)[right =of P1]{$P_2$};
\node[rectangle,minimum size=1cm,scale=0.9] (I)[right =of P2]{$\cdots$};
\node[rectangle,draw,minimum size=1cm,scale=0.9] (P3)[right =of I]{$P_N$};
\node[cloud, draw,cloud puffs=14.7,cloud ignores aspect, cloud puff arc=120,minimum width=4.cm,minimum height=2cm,node distance=0.0cm, aspect=2, inner ysep=1em] (c1) [right =of I1] {Network};
\node[rectangle,draw,minimum size=1cm,scale=0.9,node distance=.4cm] (E1)[right =of c1]{$\begin{array}{c}M\\ \mbox{Monitoring or} \\ \mbox{ Estimating Unit} \end{array}$};
\tikzset{every node/.style={fill=white}}
\path (P1)  edge [->,double=black] (c1)
      (P2)  edge [->,double=black] (c1)
      (P3)  edge [->,double=black] (c1);
\path (c1)  edge [->,double=black] (E1);
\node[circle,fill={rgb:black,1;white,2},draw,minimum size=3.7cm,scale=0.9,fill opacity=0.2,node distance=1.66cm] (Z2)[above right=of Z1]{};
\node[circle,minimum size=0cm,node distance=-.27cm,fill opacity=0] (T1)[below=of Z2]{};
\draw (Z1) -- (T1);
\node[circle,minimum size=0cm,node distance=-.27cm,fill opacity=0] (T2)[left=of Z2]{};
\draw (Z1) -- (T2);
\node[rectangle,minimum size=.3cm,scale=0.9,node distance=2.03cm](P1_1)[above =of T1]{};
\node[rectangle,draw,minimum size=.3cm,scale=0.9,node distance=0.36cm](P1_2)[left =of P1_1]{$P_1^{1}$};
\node[rectangle,draw,minimum size=.3cm,scale=0.9,node distance=0.36cm](P1_3)[right =of P1_2]{$P_1^{2}$};
\node[rectangle,draw,minimum size=.3cm,scale=0.9,node distance=0.46cm](P1_4)[below right =of P1_3]{$P_1^{3}$};
\node[rectangle,draw,minimum size=.3cm,scale=0.9,node distance=0.36cm](P1_5)[below left =of P1_4]{$P_1^{4}$};
\node[rectangle,draw,minimum size=.3cm,scale=0.9,node distance=0.26cm](P1_6)[below right =of P1_2]{$P_1^{5}$};
\node[rectangle,draw,minimum size=.3cm,scale=0.9,node distance=0.32cm](P1_7)[below left =of P1_6]{$P_1^{6}$};
\path (P1_2)  edge [->,black] (P1_3)
      (P1_6)  edge [->,black] (P1_5)
      (P1_4)  edge [bend left,->,black] (P1_5)
      (P1_7)  edge [bend left,->,black] (P1_3)
      (P1_3)  edge [bend left,->,black] (P1_5)
      (P1_4)  edge [->,black,out=80,in=55] (P1_2)
      (P1_5)  edge [->,black,out=-120,in=-60] (P1_7)
      (P1_7)  edge [->,black,out=150,in=-160] (P1_2)
      (P1_3)  edge [->,black,out=-30,in=130] (P1_4)
      (P1_6)  edge [->,black] (P1_7);
\end{tikzpicture}
\caption{\label{figure:network} Schematic diagram of the networked monitoring problem.}
\end{figure}

Compressive sensing (or solving under-determined linear set of equations) aims at reconstructing a high dimensional source vectors from low dimensional measurement vectors; see~\cite{08:Candes,06:Candes2,07:Candes,07:Tropp,08:Blumensath,09:Dai,09:Elad,10:Protter} among others for a survey of the results in the compressive sensing and possible algorithms for the signal reconstruction. For this purpose, one important condition is that the source vector possesses a sparse representation; i.e., a representation with most of its components likely to be zero. This representation is often performed using a transform known as the sparsifying basis. A careful design of the sparsifying basis is a crucial step in order to obtain potential performance gains using the compressive sensing. Common pre-defined bases are wavelet transforms, discrete cosine transforms, and curvelets. However, it has been shown in~\cite{06:Aharon,rosenblum2010dictionary} that the use of optimal bases (with respect to a specific criterion), rather than using pre-defined bases, can improve the performance of compressive sensing algorithms. Furthermore, when dealing with dynamical systems, we need to find a sparsifying basis that keeps the system state sparse at all time steps (which might not be possible using pre-defined bases). We use sparsifying bases to represent the state vector of each system in a sparse manner.

In this paper, we focus on system-estimator architecture in Figure~\ref{figure:network}. We let the system and the estimator agree on a sparsifying basis using a common history of the system state measurements (available to both of them). Then, the system encodes its state vector using the agreed basis and transmits the encoded information to the estimator (which should contain far fewer measurements than the original one\footnote{\blue{Note that if we do not have a reasonable model for the system or if many stochastic disturbances are acting on the system simultaneously, we probably cannot find a sparse representation for the state vector. Then the state vector of the system is not compressible using the approach of this paper.}}\hspace{-.03in}). Given that the recovery error is small enough, the system and its estimator can update their basis at each time step separately while managing to keep the difference negligible. Then, we find an upper bound for the recovery error based on the starting point \Red{(i.e., the initial state of the system, the initialization of the estimator, etc.)} and the modeling error of the system. Now, if each system in Figure~\ref{figure:network} uses this protocol for transmitting its state vector, resources of the communication network would be saved in comparison to the situation where the systems transmit their state vectors completely.

There have been other studies in using compressive sensing for  networked control and estimation~\cite{13:Dai,wakin2010observability,5991014,nagahara2012compressive1,nagahara2012compressive2,6426647}.
For instance, the authors in~\cite{13:Dai} studied the observability (i.e., recovering the initial state) of linear systems with a sparse initial state. In~\cite{5991014}, the authors proposed a method using the compressive sensing to close the feedback loop. However, to the best of our knowledge, the problem of learning an optimal sparsifying basis in context of networked estimation or monitoring has not been considered.

The rest of the paper is organized as follows. In Section~\ref{sec:model}, we present the problem formulation. In Section~\ref{sec:procedure}, we introduce the basis prediction procedure as well as the encoding and decoding algorithms that the system and the estimator utilize in order to communicate the state vector over the shared medium. We calculate an upper bound for the estimation error in Section~\ref{sec:results}. Finally, we demonstrate the results on a numerical example in Section~\ref{sec:example} and conclude the paper in Section~\ref{sec:conclusions}.

\Red{Notation:} We use \blue{$\mathbb{Z}$, $\mathbb{N}$, and $\mathbb{R}$ to denote the sets of integers, integers greater than or equal to one (i.e., natural numbers), and reals, respectively. Let us also define $\mathbb{N}_0=\mathbb{N}\cup\{0\}$. } We use calligraphic roman letters\blue{, such as $\mathcal{R}$, } to show all other sets. \blue{We use $|\mathcal{R}|$ to denote the cardinality of any $\mathcal{R}$. } Matrices are denoted by capital roman letters, such as $A$ and $E$. For any $1\leq i\leq m$, $A_i$ denotes $i$-th column of the matrix $A\in\mathbb{R}^{n\times m}$. For any \Red{$n,m\in\mathbb{N}$,} we define $\modd(n,m)=n-\lfloor n/m\rfloor m$. \blue{For any matrix $A\in\mathbb{R}^{n\times m}$, we define the notation $\spans(A)\subseteq\mathbb{R}^n$ to denote a set which is composed of all linear combinations of the columns of $A$. We define $\ker(A)$ as the space of all vectors $x$ such that $Ax=0$. } \blue{With slight abuse of notation, the $\ell_2$-norm of a vector and spectral norm of a matrix are both denoted by~$\|\cdot\|_2$. }

\section{System Model} \label{sec:model}
\blue{In the rest of this paper, we focus only on one of the systems in Figure~\ref{figure:network} as all the results can be readily extended to other systems. For illustrative purposes, we assume that $P_1$ is the system that we consider. } This system is indeed an interconnected dynamical system composed of $L$ physically interacting discrete-time linear time-invariant subsystems (e.g., $P_1^{\ell}$, $1\leq \ell\leq L$, in the magnified part of Figure~\ref{figure:network}). Subsystem~$\ell$, $1\leq \ell\leq L$, \Red{at time step $k\in\mathbb{N}_0$,} is described in state-space form by \begin{equation}\label{eqn:systemmodel}
x_\ell(k+1)=\sum_{j=1}^L (A_{\ell j}+\tilde{A}_{\ell j}) x_\ell(k)+E_\ell w_\ell(k);\; x_\ell(0)=x_0^{(\ell)},
\end{equation}
where $x_\ell(k)\in\mathbb{R}^{n_\ell}$ and $w_\ell(k)\in\mathbb{R}^{p_\ell}$ are its state vector and exogenous input, respectively. \blue{In~\eqref{eqn:systemmodel}, $A_{\ell j}$ is the nominal model and $\tilde{A}_{\ell j}$ is the deviation from this nominal model. } Let us denote the augmented system by
\begin{equation}\label{eqn:entiresystemmodel}
x(k+1)=(A+\tilde{A})x(k)+Ew(k); \; x(0)=x_0,
\end{equation}
where
$$
x(k)=\matrix{c}{x_1(k) \\  \vdots \\ x_L(k)}, \hspace{.1in} w(k)=\matrix{c}{w_1(k) \\  \vdots \\ w_L(k)}, \hspace{.1in} x_0=\matrix{c}{x_0^{(1)} \\  \vdots \\ x_0^{(L)}},
$$
and $E=\diag(E_1,\dots,E_L)$, and
$$
A=\matrix{ccc}{A_{11} & \cdots & A_{1L} \\  \vdots & \ddots & \vdots \\ A_{L1} & \cdots & A_{LL} }, \;\;
\tilde{A}=\matrix{ccc}{\tilde{A}_{11} & \cdots & \tilde{A}_{1L} \\  \vdots & \ddots & \vdots \\ \tilde{A}_{L1} & \cdots & \tilde{A}_{LL} }. 
$$
\blue{In this definition, we have $x(k)\in\mathbb{R}^n$, where $n=\sum_{\ell=1}^L n_\ell$. } We are interested in estimating the system state from the measurement vector
\begin{equation}
y(k)=C(k)x(k)\in\mathbb{R}^{p(k)},
\end{equation}
where $p(k)\in\mathbb{N}$ is the observation vector dimension and $C(k)$ is the observation matrix. We will discuss the observation matrix (and its properties, e.g., the observation vector dimension) in detail in the following section.
With this model in hand, we are ready to describe the encoding and decoding algorithms \Red{used by the system and the estimator.}

\section{Transmitter and Receiver Algorithm} \label{sec:procedure}
The system and the estimator use the block diagram in Figure~\ref{figure:1} to transmit and to receive the state measurements. In the remainder of this section, we explain each block individually.

\subsubsection{System:}
The first block illustrates the system dynamics. The output of this block is the state vector of the system which we are planning to encode and transmit across the communication network.

\subsubsection{Calculating the Basis at the Transmitter:}
At time-step $k\in\mathbb{N}_0$, the transmitter solves the optimization problem
\begin{equation}
\begin{split}
(Y(k),s(k),D(k))\in\argmin_{\footnotesize
\begin{array}{c}
Y\in\mathbb{R}^{m\times H} \\ s\in\mathbb{N} \\
D\in\mathbb{R}^{n\times m}
\end{array}
} &\; \|X(k)-DY\|_F, \\ \hspace{0.6in}\mathrm{subject\;to}\hspace{0.07in} &\; \|Y_i\|_0\leq s, \\ & \hspace{0.1in} \forall i\in\{1,\dots,H\},
\end{split}
\end{equation}
where $H=H_b+H_f$ with backward and forward horizons $H_b,H_f\in\mathbb{N}_0$, and
$$
X(k)=\matrix{cccccc}{\hspace{-.03in}x(k-H_b) & \cdots& x(k-1)& Ax(k-1) & \cdots & A^{H_f}x(k-1)}.
$$
This optimization problem is NP-hard in general~\cite{rosenblum2010dictionary}. Therefore, we use the sub-optimal greedy algorithm in Procedure~\ref{alg:1} for \Red{approximately} solving this problem; see~\cite{760624,rosenblum2010dictionary} for a discussion on this algorithm. This algorithm uses orthogonal matching pursuit \blue{(denoted by OMP in Procedure~\ref{alg:1} and introduced in Procedure~\ref{alg:OMP}) } which is a signal reconstruction algorithm in the compressive sensing literature; see~\cite{342465,07:Tropp} for a discussion on orthogonal matching pursuit.

\begin{remark} Note that the solution to this optimization problem is definitely not unique. This is indeed true since $Y_i$ are sparse vectors and hence, the columns of $D(k)$ that do not belong to $\bigcap_{i=1}^H \supp(Y_i)$ can be changed arbitrarily without affecting the outcome of the optimization, \Red{where for any $y\in\mathbb{R}^m$, $\supp(y)$ denotes the set of all indexes such that the components of $y$ are non-zero.}
\end{remark}

\begin{algorithm}[t]
\begin{small}
\caption{Greedy algorithm for basis optimization.}
\label{alg:1}
\begin{algorithmic}[1]
\REQUIRE $X(k)\in\mathbb{R}^{n\times H}$, $\epsilon>0$, $\varepsilon>0$, $\theta>0$.
\ENSURE $Y(k)\in\mathbb{R}^{m\times H}$, $D(k)\in\mathbb{R}^{n\times m}$, \blue{$s(k)\in\mathbb{N}$}.
\\ \hspace{-.24in } \textbf{Initialization:} Fix $D_{\mathrm{old}}=0$. Whenever $m\leq H$, use $D=[X_1(k)\; \dots\;X_m(k)]$; otherwise, use $D=[X(k)\;R]$ where $R\in\mathbb{R}^{n\times(H-m)}$ is a random matrix.
\WHILE{$\|D-D_{\mathrm{old}}\|\geq \theta$}
\STATE $D_{\mathrm{old}}\leftarrow D$.
\FOR{$i=1,\dots,H$}
\STATE $Y_i\leftarrow \mathrm{OMP}(D,X_i(k),\varepsilon)$.
\ENDFOR
\STATE $D\leftarrow XY^\top(YY^\top+\epsilon I)^{-1}$ \blue{(Note $\epsilon>0$ allows the expression to become invertible)}.
\ENDWHILE
\STATE $D(k)$ is the normalized version of $D$.
\end{algorithmic}
\end{small}
\end{algorithm}

\begin{algorithm}[t]
\blue{
\begin{small}
\caption{Orthogonal Matching Pursuit (OMP).}
\label{alg:OMP}
\begin{algorithmic}[1]
\REQUIRE $D\in\mathbb{R}^{n\times m}$, $x\in\mathbb{R}^n$, $\varepsilon>0$.
\ENSURE $y\in\mathbb{R}^m$.
\\ \hspace{-.246in } \textbf{Initialization:} Fix $r=x$, $y=0$, $\mathcal{S}=\emptyset$.
\WHILE{$\|r\|_2\geq \varepsilon$}
\STATE $i^*\leftarrow \argmax_{i\in\mathcal{S}^{c}} D_i^\top r$.
\STATE $\mathcal{S} \leftarrow \mathcal{S}\cup \{i^*\}$.
\STATE $s\leftarrow |\mathcal{S}|$.
\STATE Set $y\in\mathbb{R}^m$ such that $y_\mathcal{S}\leftarrow \argmin_{y'\in\mathbb{R}^{s}} \|x-D_{s}y'\|_2$ (where $D_{\mathcal{S}}\in\mathbb{R}^{n\times s}$ is a submatrix of $D$ generated by keeping all its columns belonging to $\mathcal{S}$) and $y_{\mathcal{S}^c}=0$ (where $\mathcal{S}^c$ denotes the complement of $\mathcal{S}$).
\STATE $r\leftarrow x-Dy$.
\ENDWHILE
\end{algorithmic}
\end{small}}
\end{algorithm}

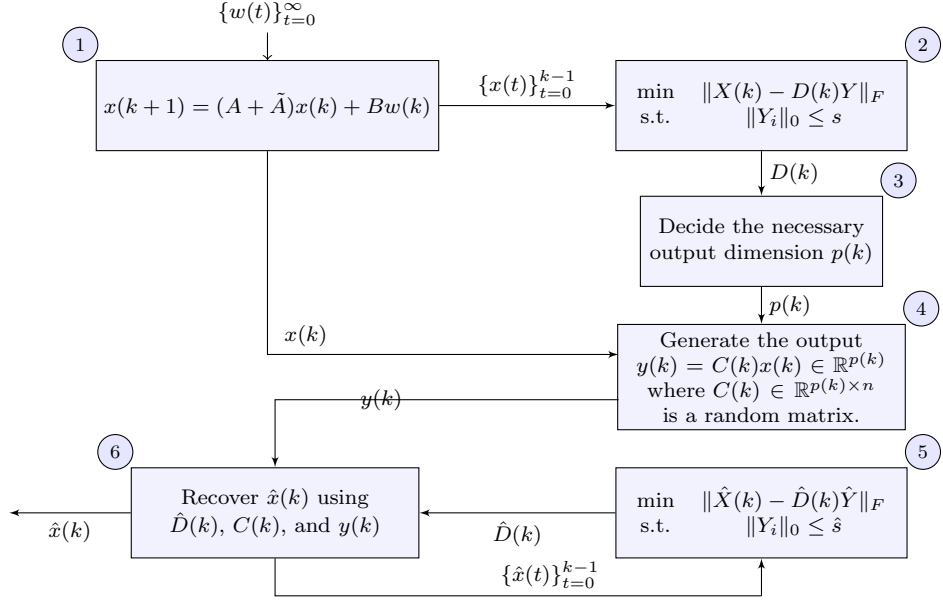
\begin{figure*}
\centering
\footnotesize
\tikzstyle{block}=[draw,fill=blue!5,rectangle,minimum height=4em, minimum width=6em]
\tikzstyle{label1}=[draw,fill=white,rectangle]
\tikzstyle{input}=[coordinate]
\tikzstyle{output}=[coordinate]
\tikzstyle{pinstyle}=[pin edge={to-,thin,black}]
\begin{tikzpicture}[auto,>=latex']
\node[block,name=system,pin={[pinstyle]above:$\{w(t)\}_{t=0}^\infty$}] {$x(k+1)=(A+\tilde{A})x(k)+Bw(k)$};
\node[circle,draw,minimum size=.3cm,scale=0.9,node distance=0.05cm,fill=blue!10](Number1)[above left =of system]{$1$};
\node[block,node distance=6.5cm,right of=system,name=Dk] {$\begin{array}{cc} \min & \|X(k)-D(k)Y\|_F \\ \mathrm{s.t.} & \|Y_i\|_0\leq s \end{array}$};
\node[circle,draw,minimum size=.3cm,scale=0.9,node distance=0.05cm,fill=blue!10](Number2)[above right =of Dk]{$2$};
\node [block,node distance=1.8cm,below of=Dk,name=pk,text width=10em,text centered]{Decide the necessary \\ output dimension $p(k)$};
\node[circle,draw,minimum size=.3cm,scale=0.9,node distance=0.05cm,fill=blue!10](Number3)[above right =of pk]{$3$};
\node [block,node distance=1.8cm,below of=pk,name=ck,text width=12em,text centered]{Generate the output \\ $y(k)=C(k)x(k)\in\mathbb{R}^{p(k)}$ where $C(k)\in\mathbb{R}^{p(k)\times n}$ \\ is a random matrix. };
\node[circle,draw,minimum size=.3cm,scale=0.9,node distance=0.05cm,fill=blue!10](Number4)[above right =of ck]{$4$};
\node [block,node distance=1.8cm,below of=ck,name=hatDk,text centered]{$\begin{array}{cc} \min & \|\hat{X}(k)-\hat{D}(k)\hat{Y}\|_F \\ \mathrm{s.t.} & \|{Y}_i\|_0\leq \hat{s} \end{array}$};
\node[circle,draw,minimum size=.3cm,scale=0.9,node distance=0.05cm,fill=blue!10](Number5)[above right =of hatDk]{$5$};
\node [block,node distance=6.4cm,left of=hatDk,name=hatxk,text width=12em,text centered]{Recover $\hat{x}(k)$ using \\ $\hat{D}(k)$, $C(k)$, and $y(k)$};
\node[circle,draw,minimum size=.3cm,scale=0.9,node distance=0.05cm,fill=blue!10](Number6)[above left =of hatxk]{$6$};
\node [output,node distance=3.5cm,left of=hatxk](output){};
\node [output,node distance=1.1cm,below of=hatDk](output1){};
\node [output,node distance=1.88cm,left of=ck](output2){};
\node [output,node distance=0.3cm,above of=output2](output3){};
\node [output,node distance=0.3cm,below of=output2](output4){};
\draw [draw,->](system)--node{$\{x(t)\}_{t=0}^{k-1}$}(Dk);
\draw [draw,->](Dk)--node{$D(k)$}(pk);
\draw [draw,->](pk)--node{$p(k)$}(ck);
\draw [draw,->](system)|-node{\hspace{.4in}$x(k)$}(output3);
\draw [draw,->](hatDk)--node{$\hat{D}(k)$}(hatxk);
\draw [draw,->](hatxk)--node{$\hat{x}(k)$}(output);
\draw [draw,-](hatxk)|-node{\hspace{2.85in}$\{\hat{x}(t)\}_{t=0}^{k-1}$}(output1);
\draw [draw,->](output1)--node{}(hatDk);
\draw [draw,->](output4)-|node[label=below:\hspace{.4in}$y(k)$,name=y]{}(hatxk);
\end{tikzpicture}
\caption{\label{figure:1} The block diagram of the transmission protocol using the basis optimization.} 
\end{figure*}

Let $\mathcal{I}(k)$ denote the set of all indices $1\leq i\leq m$ such that $D_i(k)\neq 0$. The sparsity is given by $s(k)=|\mathcal{I}(k)|$. We define $D_{\mathcal{I}(k)}(k)\in\mathbb{R}^{n\times s(k)}$ as a submatrix of $D(k)\in\mathbb{R}^{n\times m}$ generated by all its columns belonging to the set $\mathcal{I}(k)$. In addition, for any $z\in\mathbb{R}^m$, we define $z_{\mathcal{I}(k)}\in\mathbb{R}^{s(k)}$ to be the vector composed of all $z_i$ such that $i\in\mathcal{I}(k)$. Clearly, $D(k)z=D_{\mathcal{I}(k)}(k)z_\mathcal{I}$ for all $z\in\mathbb{R}^m$. Now, let us define a vector $z(k)\in\mathbb{R}^m$ such that $z_{\mathcal{I}(k)^c}(k)=0$, and
\begin{equation}
z_{\mathcal{I}(k)}(k)=\argmin_{v\in\mathbb{R}^{s(k)}} \|x(k)-D_{\mathcal{I}(k)}(k)v\|_2.
\end{equation}
Note that $z(k)$ is a $s(k)$-sparse representation of $x(k)$ using basis $D(k)$. \Red{The error caused by this} representation is $e(k)=x(k)-D(k)z(k)$. Later, we use the notation $\delta_s(k)=\|e(k)\|_2$ to denote the norm of this representation error.
\begin{remark} Notice that $\spans(\tilde{A})\subset \spans(X(k))$ and $w(k-1)\in\spans(X(k))$ implies that $e(k)=0$. However, in general, $x(k)$ might not be $s(k)$-sparse representable using basis $D(k)$. This is indeed true because we construct the basis assuming that $w(k-1)=\cdots=w(k+H_f-2)=0$ and $\tilde{A}=0$ whenever $H_f>0$.
\end{remark}

\subsubsection{Calculating the Output Dimension:}
In this block, we set the number of outputs $p(k)$ that one requires in order to recover the state vector in the estimator. Note that in a perfect situation, we need to fix this number equal to the sparsity level $s(k)$. However, in practice, $p(k)$ might vary considering different recovery algorithms in the receiver (to ensure the stability of the algorithm or to decrease the recovery error). For the moment, it suffices to fix this number as $p(k)=\mathcal{O}(s(k))$.

\begin{algorithm}[t]
\begin{small}
\caption{Numerical procedure for constructing the output vector. }
\label{alg:2}
\begin{algorithmic}[1]
\REQUIRE Matrices $C_\ell(k)$ for $1\leq \ell\leq L$.
\ENSURE $y(k)$
\FOR{$t=0,\dots,T-1$}
\FOR{$\ell\in\mathcal{C}_{T-t}$}
\IF{$t=0$}
\STATE Subsystem $P_\ell$ calculates~$z_\ell(0)=C_\ell(k) x_\ell(k)$.
\ELSE
\STATE Subsystem $P_\ell$ calculates~$z_\ell(t)=C_\ell(k) x_\ell(k)+ \sum_{i\in\mathcal{N}_\ell}z_i(t-1)$, where $\mathcal{N}_\ell=\{j\;|\;(P_j,P_\ell)\in\mathcal{E}\}$.
\ENDIF
\STATE Subsystem $P_\ell$ transmits~$z_\ell(t)$ to subsystem~$P_i$ such that $(P_\ell,P_i)\in\mathcal{E}$.
\ENDFOR
\ENDFOR
\STATE The estimator calculates $y(k)=\sum_{i\in\mathcal{C}_1}z_i(T-1)$.
\end{algorithmic}
\end{small}
\end{algorithm}

\subsubsection{Constructing the Output Vector:}
In this section, we develop a distributed algorithm for constructing the observation vector $y(k)=C(k)x(k)$ with a stochastic observation matrix $C(k)$ where its entries are identically and independently distributed Gaussian random variables~$\mathcal{N}(0,1)$. Let us define the communication graph.

\begin{definition} The communication graph $\mathcal{G}=(\mathcal{V},\mathcal{E})$ is a directed graph with the vertex set $\mathcal{V}=\{0,1,\dots,L\}$, where $0$ denotes the estimator node and $\ell$ denotes subsystem~$\ell$ for $1\leq \ell \leq L$. An edge such as $(i,j)\in\mathcal{E}$ with $1\leq i,j\leq L$, shows that subsystem $j$ can receive the information transmitted by subsystem $i$ and an edge such as $(i,0)\in\mathcal{E}$ with $1\leq i\leq L$, shows that the estimator can receive the information transmitted by subsystems $i$.
\end{definition}

We make the following standing assumption concerning the communication graph.

\begin{assumption} \label{assum:2} The communication graph $\mathcal{G}=(\mathcal{V},\mathcal{E})$ is an acyclic directed graph. Furthermore, for each $1\leq \ell\leq L$, there exists \emph{exactly one path} that connects node $\ell$ to node $0$.
\end{assumption}

\begin{figure}[t]
\centering
\begin{tikzpicture}[>=stealth',shorten >=2pt,node distance=0.95cm,initial/.style={}]
  \node[state,minimum size=0.1cm,scale=0.9] (P1)                       {$1$};
  \node[state,minimum size=0.1cm,scale=0.9] (P7) [left        =of P1]  {$6$};
  \node[state,minimum size=0.1cm,scale=0.9] (P2) [right       =of P1]  {$2$};
  \node[state,minimum size=0.1cm,scale=0.9] (P3) [right       =of P2]  {$3$};
  \node[state,minimum size=0.1cm,scale=0.9] (P4) [below right =of P1]  {$4$};
  \node[state,minimum size=0.1cm,scale=0.9] (P5) [right       =of P4]  {$5$};
  \node[state,minimum size=0.1cm,scale=0.9] (P6) [below right =of P4]  {$0$};
\tikzset{mystyle/.style={->,double=black}}
\tikzset{every node/.style={fill=white}}
\path (P1)  edge [mystyle] (P4)
      (P7)  edge [mystyle] (P1)
      (P2)  edge [mystyle] (P4)
      (P4)  edge [mystyle] (P6)
      (P5)  edge [mystyle] (P6);
\tikzset{mystyle/.style={->,relative=false,in=100,out=-160,double=black}}
\path (P3)  edge [mystyle] (P6);
\end{tikzpicture}
\caption{\label{graph1} An example of the communication network $\mathcal{G}$.}
\end{figure}
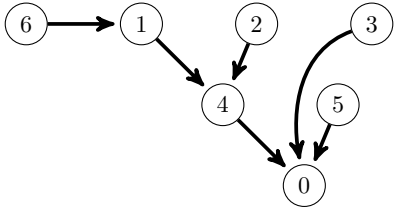

Let us again consider the magnified part of the schematic diagram in Figure~\ref{figure:network}. As illustrated there, six subsystems form $P_1$ (which is trying to transmit its state vector to the remote monitoring unit $M$). The directed graph in Figure~\ref{graph1} shows an example of the communication graph among these six subsystems and the monitoring unit. Clearly, this communication graph satisfies Assumption~\ref{assum:2}.

We define the partition $\mathcal{C}_1 \cup \dots \cup \mathcal{C}_T=\{1,\dots,L\}$, where $T$ is the maximum length of a path in the communication graph $\mathcal{G}$. For each $1\leq t\leq T$, we define $\mathcal{C}_t$ as the set of all vertices $j$ such that there is a path of length $t$ that connects $j$ to $0$ over the communication graph $\mathcal{G}$. Note that due to Assumption~\ref{assum:2}, the sets $\mathcal{C}_t$ for $1\leq t\leq T$, are unique and well-defined (i.e., $\mathcal{C}_{t_1}\cap \mathcal{C}_{t_2}=\phi$ for $1\leq t_1\neq t_2\leq T$).

Now, let us construct random matrices $C_\ell(k)\in\mathbb{R}^{p(k)\times n_\ell}$, $1\leq \ell \leq L$, where their entries are identically and independently distributed Gaussian random variables with probability distribution~$\mathcal{N}(0,1)$. Let us assume that the communication delay is negligible in comparison to the system dynamics. Hence, we can have as many communication rounds as we desire in just one time step. Procedure~\ref{alg:2} describes a numerical procedure that the subsystems follow to construct the output vector $y(k)=C(k)x(k)$ distributedly at time step $k\in\mathbb{N}_0$. 
Following the steps in Procedure~\ref{alg:2}, it is easy to see that 
\begin{equation} \label{eqn:long:4}
\begin{split}
y(k)&=\sum_{{i_1}\in\mathcal{C}_1}C_{i_1}(k) x_{i_1}(k)+
\sum_{{i_1}\in\mathcal{C}_1}\sum_{{i_2}\in\mathcal{N}_{i_1}}C_{i_2}(k) x_{i_2}(k)+\cdots+
\sum_{{i_1}\in\mathcal{C}_1}\cdots\hspace{-.15in}\sum_{{i_T}\in \mathcal{N}_{i_{T-1}}} \hspace{-.15in}C_{i_T}(k) x_{i_T}(k)
\\ & =\sum_{{i_1}\in\mathcal{C}_1}C_{i_1}(k) x_{i_1}(k)+
\sum_{{i_2}\in\mathcal{C}_2}C_{i_2}(k) x_{i_2}(k)+\cdots+
\sum_{{i_T}\in\mathcal{C}_T}C_{i_T}(k) x_{i_T}(k)
\\ &=\sum_{\ell=1}^L C_\ell(k) x_\ell(k),
\end{split}
\end{equation}
where the second equality is due to the fact that each subsystem is connected to node~$0$ only through one path (see Assumption~\ref{assum:2}) which proves that $\mathcal{C}_{t_1}\cap \mathcal{C}_{t_2}=\phi$ for $1\leq t_1\neq t_2\leq T$.

\begin{remark} As we will see later, the recovery algorithm requires the observation matrix $C(k)$ at each time step $k\in\mathbb{N}_0$ (which is not possible to communicate to the monitoring station). However, we can change the random matrices $C(k)$ with pseudo-random ones. Doing so, the recovery algorithm would only need the seeds of these pseudo-random number generators. Examples of such pseudo-random number generators are linear feedback shift register~\cite{goresky2012algebraic} and complementary-multiply-with-carry~\cite{Marsaglia1991}.
\end{remark}

\subsubsection{Calculating the Basis at the Receiver:}
At time step $k\in\mathbb{N}_0$, the receiver solves the optimization problem
\begin{equation}
\begin{split}
(\hat{Y}(k),\hat{s}(k),\hat{D}(k))\in\argmin_{\footnotesize
\begin{array}{c}
Y\in\mathbb{R}^{m\times H} \\ s\in\mathbb{N} \\
D\in\mathbb{R}^{n\times m}
\end{array}
} &\; \|\hat{X}(k)-DY\|_F, \\ \hspace{0.8in}\mathrm{s.t.} &\; \|Y_i\|_0\leq s,  \\ \hspace{1.5in} &\;\forall i\in\{1,\dots,H\},
\end{split}
\end{equation}
where $H$ is equal to the one in the transmitter and
$$
\hat{X}(k)=\matrix{cccccc}{\hat{x}(k-H_b) & \cdots & \hat{x}(k-1)& A\hat{x}(k-1) & \cdots & A^{H_f}\hat{x}(k-1)}.
$$

\subsubsection{Recovering the State Vector:}
Let us use $\hat{\mathcal{I}}(k)$ to denote the set of all indices $1\leq i\leq m$ such that $\hat{D}_i(k)\neq 0$. The sparsity is given by $\hat{s}(k)=|\hat{\mathcal{I}}(k)|$. Now, we can define a vector $\hat{z}(k)\in\mathbb{R}^m$ such that $\hat{z}_{\hat{\mathcal{I}}(k)^c}(k)=0$, and
\begin{equation}
\hat{z}_{\hat{\mathcal{I}}(k)}(k)=\argmin_{v\in\mathbb{R}^{\hat{s}(k)}} \|y(k)-C(k)\hat{D}_{\hat{\mathcal{I}}(k)}(k)v\|_2.
\end{equation}
The state estimate at the receiver is $\hat{x}(k)=\hat{D}(k)\hat{z}(k)$.

\section{Performance Bound}\label{sec:results}
In this section, we find an upper-bound for the estimation error. First, we need to prove the following useful lemmas.

\begin{lemma} \label{lemma:1} For any $z\in\mathbb{R}^m$, there exists $\alpha\in\mathbb{R}^{H}$ (with $H=H_f+H_b$) such that $D(k)z=D(k)Y(k)\alpha$.
\end{lemma}

\begin{proof} If $\min_{(\alpha_i)_{i=1}^{H}}\|z- \sum_{i=1}^H\alpha_iY_i(k)\|_2=0$, the equality is trivially satisfied. Hence, without loss of generality, we assume that $\min_{(\alpha_i)_{i=1}^{H}}\|z-\sum_{i=1}^H \alpha_iY_i(k)\|_2>0$. Let us define the notation
\begin{equation} \label{eqn:alpha*}
(\alpha_i^*)_{i=1}^{H}\in\argmin_{(\alpha_i)_{i=1}^{H}}\left\|z- \sum_{i=1}^H\alpha_iY_i(k)\right\|_2,
\end{equation}
and subsequently, $z_0=z-\sum_{i=1}^H \alpha_iY_i(k)$. We prove that $Y_i(k)^\top z_0=0$ for all $1\leq i\leq H$. Assume that this is not true. Therefore, there exists $j\in\{1,\dots,H\}$ such that $Y_j(k)^\top z_0\neq 0$. Now, we define $\bar{\alpha}_i=\alpha_i^*$ for all $i\neq j$ and
$$
\bar{\alpha}_j=\alpha^*_j+\frac{Y_j(k)^\top z_0}{Y_j(k)^\top Y_j(k)}.
$$
It is easy to see that
\begin{equation*}
\begin{split}
Y_j(k)^\top \left(z-\sum_{i=1}^H \bar{\alpha}_iY_i(k)\right)&=
Y_j(k)^\top \left(z-\sum_{i=1}^H \alpha_i^*Y_i(k)-\frac{Y_j(k)^\top z_0}{Y_j(k)^\top Y_j(k)}Y_j(k)\right)\\ &= Y_j(k)^\top \left(z_0-\frac{Y_j(k)^\top z_0}{Y_j(k)^\top Y_j(k)}Y_j(k)\right)=0.
\end{split}
\end{equation*}
In addition, we have
\begin{equation*}
\begin{split}
\bigg\|z-\sum_{i=1}^H \alpha_i^*Y_i(k)\bigg\|_2 &=\left\|\frac{Y_j(k)^\top z_0}{Y_j(k)^\top Y_j(k)}Y_j(k)+z-\sum_{i=1}^H \bar{\alpha}_iY_i(k)\right\|_2 \\ &=\left\|z-\sum_{i=1}^H \bar{\alpha}_iY_i(k)\right\|_2+\left\|\frac{Y_j(k)^\top z_0}{Y_j(k)^\top Y_j(k)}Y_j(k)\right\|_2,
\end{split}
\end{equation*}
where the second equality holds due to the fact that $Y_j(k)$ is \Red{orthogonal} to $z-\sum_{i=1}^H \bar{\alpha}_iY_i(k)$.
Therefore, notating that $\|(Y_j(k)^\top z_0)/(Y_j(k)^\top Y_j(k))Y_j(k)\|_2>0$, we get
$$
\left\|z-\sum_{i=1}^H \alpha_i^*Y_i(k)\right\|_2>\left\|z-\sum_{i=1}^H \bar{\alpha}_iY_i(k)\right\|_2,
$$
which contradicts~\eqref{eqn:alpha*}. Therefore, we know that $Y_i(k)^\top z_0=0$ for all $1\leq i\leq H$ and as a result,
\begin{equation*}
\begin{split}
D(k)z&=X(\epsilon I+Y(k)^\top Y(k))^{-1}Y(k)^\top z
\\&=X(\epsilon I+Y(k)^\top Y(k))^{-1}Y(k)^\top \left( z_0+\sum_{i=1}^H \alpha_i^*Y_i(k)\right)
\\&=X(\epsilon I+Y(k)^\top Y(k))^{-1}Y(k)^\top \left(\sum_{i=1}^H \alpha_i^*Y_i(k)\right)
\\&=X(\epsilon I+Y(k)^\top Y(k))^{-1}Y(k)^\top Y(k)\alpha^*,
\\&=D(k) Y(k)\alpha^*,
\end{split}
\end{equation*}
where the first equality and the last equality both follow from the identity \begin{equation} \label{eqn:long:1}
\begin{split}
D(k)&=X(k)Y(k)^\top(Y(k)Y(k)^\top+\epsilon I)^{-1}
\\&=X(k)Y(k)^\top\epsilon^{-1}(I-Y(k)(\epsilon I+Y(k)^\top Y(k))^{-1}Y(k)^\top)
\\&=X(k)\epsilon^{-1}(Y(k)^\top-Y(k)^\top Y(k)(\epsilon I+Y(k)^\top Y(k))^{-1}Y(k)^\top)
\\&=X(k)\epsilon^{-1}(Y(k)^\top-(\epsilon I+Y(k)^\top Y(k)-\epsilon I)(\epsilon I+Y(k)^\top Y(k))^{-1}Y(k)^\top)
\\&=X(k)(\epsilon I+Y(k)^\top Y(k))^{-1}Y(k)^\top.
\end{split}
\end{equation}
and the second equality follows from the fact that $Y_i(k)^\top z_0=0$ for all $1\leq i\leq H$. This completes the proof.
\end{proof}

\begin{lemma} \label{lemma:2} For any $z\in\mathbb{R}^m$, there exists $\hat{z}\in\mathbb{R}^m$ such that
\begin{equation*}
\begin{split}
\left\|D(k)z-\hat{D}(k)\hat{z}\right\|_2
 \leq&  \left( \delta_{D}(k) + \delta_{\hat{D}}(k) + \delta_{X,\hat{X}}(k) \right)  \|Y(k)^\dag z\|_2,
\end{split}
\end{equation*}
where $\delta_{D}(k)=\|D(k)Y(k)-X(k)\|_2$, $\delta_{\hat{D}}(k)=\|\hat{D}(k)\hat{Y}(k)-\hat{X}(k)\|_2$, $\delta_{X,\hat{X}}(k)=\|X(k)-\hat{X}(k)\|_2$, and $Y(k)^\dag$ denotes the Moore--Penrose pseudo-inverse of $Y(k)$.
\end{lemma}

\begin{proof}
Let us pick $\hat{z}=\hat{Y}(k)\alpha$ where $\alpha\in\mathbb{R}^H$ is chosen so that $D(k)z=D(k)Y(k)\alpha$ (see the proof of Lemma~\ref{lemma:1}). We assume that the projection of $\alpha$ to $\ker(Y(k))$ is zero. This assumption \Red{can be made} without loss of generality, as we can always use $\alpha'=\alpha-(I-Y(k)^\dag Y(k))\alpha$ instead of $\alpha$. Note that in this case, the identity $D(k)z=D(k)Y(k)\alpha'$ still holds due to~(\ref{eqn:long:1}). Now, we can easily prove the inequality
\begin{equation} \label{eqn:long:2}
\begin{split}
\left\|D(k)z-\hat{D}(k)\hat{z}\right\|_2  &=\left\|D(k)Y(k)\alpha-\hat{D}(k)\hat{Y}(k)\alpha\right\|_2
\\ &=\left\|(D(k)Y(k)-X(k))\alpha-(\hat{D}(k)\hat{Y}(k)-\hat{X}(k))\alpha +(X(k)-\hat{X}(k))\alpha\right\|_2
\\& \leq  \left(\|D(k)Y(k)-X(k)\|_2+\|\hat{D}(k)\hat{Y}(k)-\hat{X}(k)\|_2+\|X(k)-\hat{X}(k)\|_2\right) \|\alpha\|_2
\\&\leq \left( \delta_{D}(k) + \delta_{\hat{D}}(k) + \delta_{X,\hat{X}}(k) \right) \|\alpha\|_2.
\end{split}
\end{equation}
Additionally, using the fact that $Y(k)\alpha=z$ \blue{(and noting that the projection of $\alpha$ into $\ker(Y(k))$ is zero), } it is easy to see that
$$
\alpha=Y(k)^\dag z +(I-Y(k)^\dag Y(k))\alpha=Y(k)^\dag z,
$$
and as a result, $\|\alpha\|_2=\|Y(k)^\dag z\|_2$. This concludes the proof.
\end{proof}

Now, we are ready to prove the main result of this paper concerning the estimation error.

\begin{theorem} \label{tho:1} Let $\{\vartheta(k)\}_{k=0}^\infty$ be a sequence of real numbers such that $\|x(k)-\hat{x}(k)\|_2\leq \vartheta(k)$ for all $k\in\mathbb{N}_0$. Then,
\begin{equation} \label{eqn:tho:0}
\vartheta(k)\leq\delta_s(k)+\left( \delta_{D}(k) + \delta_{\hat{D}}(k) + \delta_{X,\hat{X}}(k) \right) \|Y(k)^\dag z(k)\|_2,
\end{equation}
and consequently,
\begin{equation} \label{eqn:tho:1}
\begin{split}
\vartheta(k)&\leq\delta_s(k)+\left( \delta_{D}(k) + \delta_{\hat{D}}(k) + \delta_{X,\hat{X}}(k) \right) \|Y(k)^\dag D(k)^\dag\|_2 (\delta_s(k)+\|x(k)\|_2).
\end{split}
\end{equation}
\end{theorem}

\begin{proof} The proof easily follows from the sequel of inequalities in\begin{equation} \label{eqn:long:3}
\begin{split}
\|x(k)-\hat{x}(k)\|_2&=\|x(k)-\hat{D}(k)\hat{z}(k)\|_2\\&=\|x(k)-D(k)z(k)+D(k)z(k)-\hat{D}(k)\hat{z}(k)\|_2
\\ &\leq \|x(k)-D(k)z(k)\|_2+\|D(k)z(k)-\hat{D}(k)\hat{z}(k)\|_2 \\
& \leq \delta_s(k)+\left( \delta_{D}(k) + \delta_{\hat{D}}(k) + \delta_{X,\hat{X}}(k) \right) \|Y(k)^\dag z(k)\|_2 \\
& \leq \delta_s(k)+\left( \delta_{D}(k) + \delta_{\hat{D}}(k) + \delta_{X,\hat{X}}(k) \right) \|Y(k)^\dag D(k)^\dag (x(k)-e(k))\|_2.
\end{split}
\end{equation}
\end{proof}

\blue{Notice that~\eqref{eqn:tho:1} shows that the encoding and decoding scheme at least result in a stable estimation of the state vector if the original system is stable. } To further simplify the upper bound of the estimation error, we need to prove the following lemma.

\begin{lemma} \label{lemma:3}  Let $\{\vartheta(k)\}_{k=0}^\infty$ be a sequence of real numbers such that $\|x(k)-\hat{x}(k)\|_2\leq \vartheta(k)$ for all $k\in\mathbb{N}_0$. Then,
$$
\delta_{X,\hat{X}}(k)\leq \sqrt{\frac{\|A\|_2-\|A\|_2^{H_f+2} }{1-\|A\|_2} \vartheta(k-1)^2+\sum_{i=2}^{H_b} \vartheta(k-i)^2 }.
$$
\end{lemma}

\begin{proof} For $E=X(k)-\hat{X}(k)$, we have
\begin{equation*}
\begin{split}
\|E\|_2&=\sup_{w\in\mathbb{R}^H\setminus\{0\}} \|Ew\|_2/\|w\|_2\\
 &=\sup_{w\in\mathbb{R}^H\setminus\{0\}} \left\|\sum_{i=1}^H E_iw_i\right\|_2/\|w\|_2
\\ &\leq \sup_{w\in\mathbb{R}^H\setminus\{0\}} \sum_{i=1}^H \| E_i\|_2 |w_i|/\|w\|_2
\\& \leq \sqrt{\sum_{i=1}^H \| E_i\|_2^2},
\end{split}
\end{equation*}
where the last inequality follows from the Cauchy-Schwarz inequality.
Now, by the definition of matrices $X(k)$ and $\hat{X}(k)$, we get
\begin{equation*}
\begin{split}
\|E\|_2 &\leq \sqrt{\sum_{i=1}^{H_b} \vartheta(k-i)^2 + \sum_{i=1}^{H_f} \|A^{i}(x(k-1)-\hat{x}(k-1))\|_2 }
\\ &\leq \sqrt{\sum_{i=1}^{H_b} \vartheta(k-i)^2 + \sum_{i=1}^{H_f} \|A\|_2^{i}\vartheta(k-1)^2 }.
\end{split}
\end{equation*}
Therefore,
\begin{equation*}
\begin{split}
\delta_{X,\hat{X}}(k) &\leq \sqrt{\sum_{i=1}^{H_b} \vartheta(k-i)^2 +\sum_{i=1}^{H_f} \|A\|_2^{i}\vartheta(k-1)^2 }
\\ &\leq \sqrt{\|A\|_2\frac{1-\|A\|_2^{H_f+1}}{1-\|A\|_2}\vartheta(k-1)^2 +\sum_{i=2}^{H_b} \vartheta(k-i)^2 }.
\end{split}
\end{equation*}
\end{proof}

By substituting the result of Lemma~\ref{lemma:3} into Theorem~\ref{tho:1}, we get the following two upper-bounds for the estimation error.

\begin{corollary} Let $\{\vartheta(k)\}_{k=0}^\infty$ be a sequence of real numbers such that $\|x(k)-\hat{x}(k)\|_2\leq \vartheta(k)$ for all $k\in\mathbb{N}_0$. Then,
\begin{equation*}
\begin{split}
\vartheta(k)\leq\delta_s(k)+\xi(k) \|Y(k)^\dag D(k)^\dag\|_2 (\delta_s(k)+\|x(0)\|_2),
\end{split}
\end{equation*}
where
\begin{equation*}
\begin{split}
\xi(k)\leq \delta_{D}(k) + \delta_{\hat{D}}(k) + \sqrt{\frac{\|A\|_2-\|A\|_2^{H_f+2}}{1-\|A\|_2} \vartheta(k-1)^2 +\sum_{i=1}^{H_b} \vartheta(k-i)^2 } .
\end{split}
\end{equation*}
\end{corollary}

\begin{corollary} Let $\{\vartheta(k)\}_{k=0}^\infty$ be a sequence of real numbers such that $\|x(k)-\hat{x}(k)\|_2\leq \vartheta(k)$ for all $k\in\mathbb{N}_0$. Whenever $H_b=0$, 
we get
$$
\vartheta(k)\leq \beta(k)+\gamma(k)\vartheta(k-1),
$$
where
\begin{equation*}
\begin{split}
\beta(k)=&\delta_s(k)+\left( \delta_{D}(k) + \delta_{\hat{D}}(k)\right) \|Y(k)^\dag D(k)^\dag\|_2 (\delta_s(k)+\|x(0)\|_2),
\end{split}
\end{equation*}
and
$$
\gamma(k)= \sqrt{\frac{\|A\|_2-\|A\|_2^{H_f+2}}{1-\|A\|_2}} \|Y(k)^\dag D(k)^\dag\|_2  (\delta_s(k)+\|x(0)\|_2).
$$
This would result in
$$
\vartheta(k)\leq \sum_{t=0}^k \left[\prod_{j=t+1}^k \gamma(j)\right]\beta(t) + \left[\prod_{j=0}^k \gamma(j)\right]\vartheta(0).
$$
For the special case where there exist $\Gamma,B\in\mathbb{R}$ such that $\gamma(k)\leq \Gamma<1$ and $\beta(k)\leq B$ for $k\in\mathbb{N}_0$, we get
\begin{equation*}
\begin{split}
\vartheta(k)\leq \frac{B\Gamma}{1-\Gamma} + \Gamma^{k+1}\vartheta(0).
\end{split}
\end{equation*}
\end{corollary}

\begin{figure}[!t]
\centering
\includegraphics[width=0.45\linewidth]{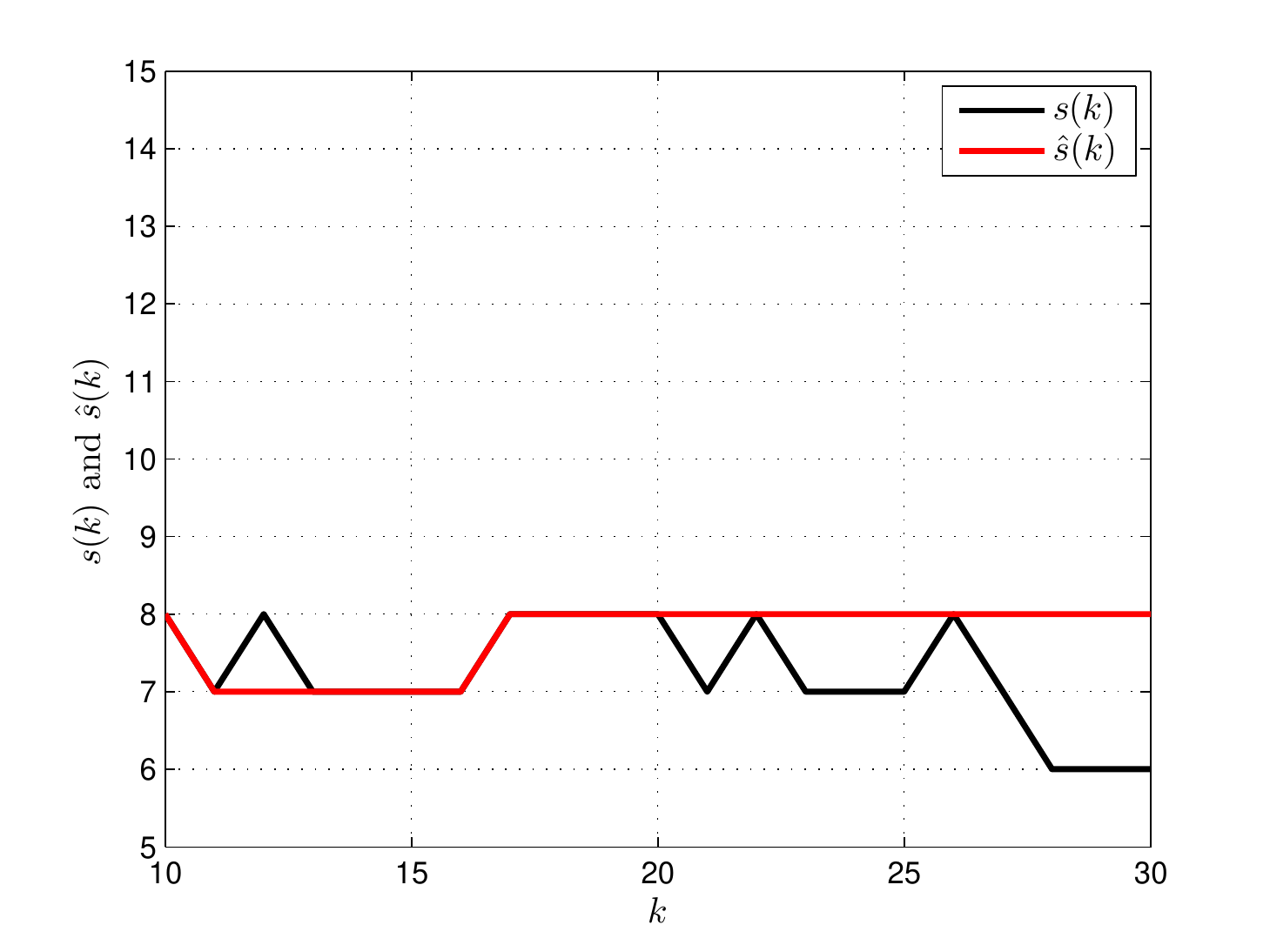}
\caption{\label{figureskhatsk} Sparsity factors $s(k)$ and $\hat{s}(k)$ as a function of time for the scenario described in Subsection~\ref{subsec:numerical:1}.}\vspace{-.06in}
\end{figure}

\begin{figure}[!t]
\centering
\includegraphics[width=0.45\linewidth]{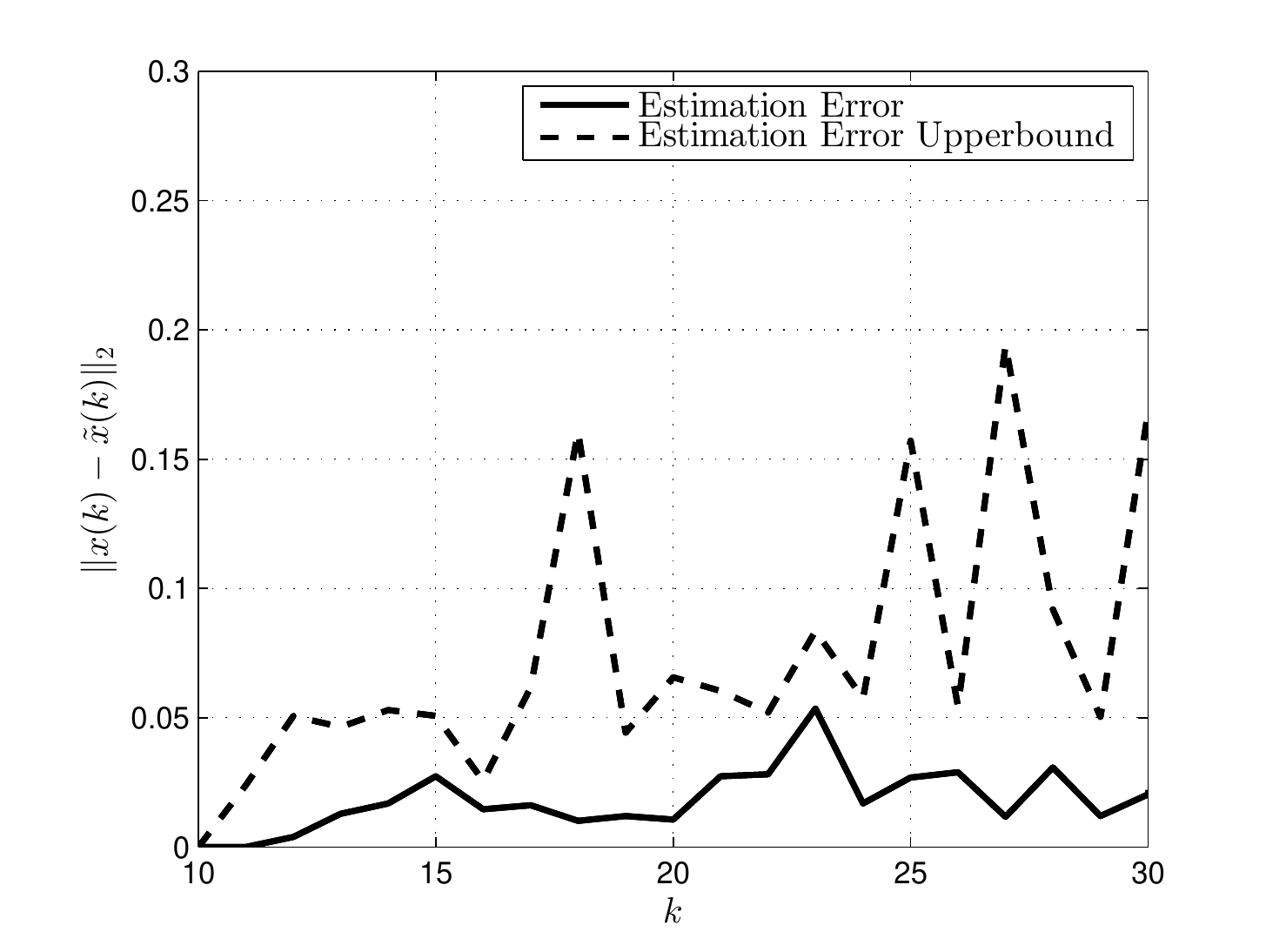}
\caption{\label{figureerror1} Estimation error and its upper bound as a function of time when $p(k)=\lceil1.3s(k)\rceil$ for all $k\in\mathbb{N}_0$.}
\end{figure}

\begin{figure}[!t]
\centering
\includegraphics[width=0.45\linewidth]{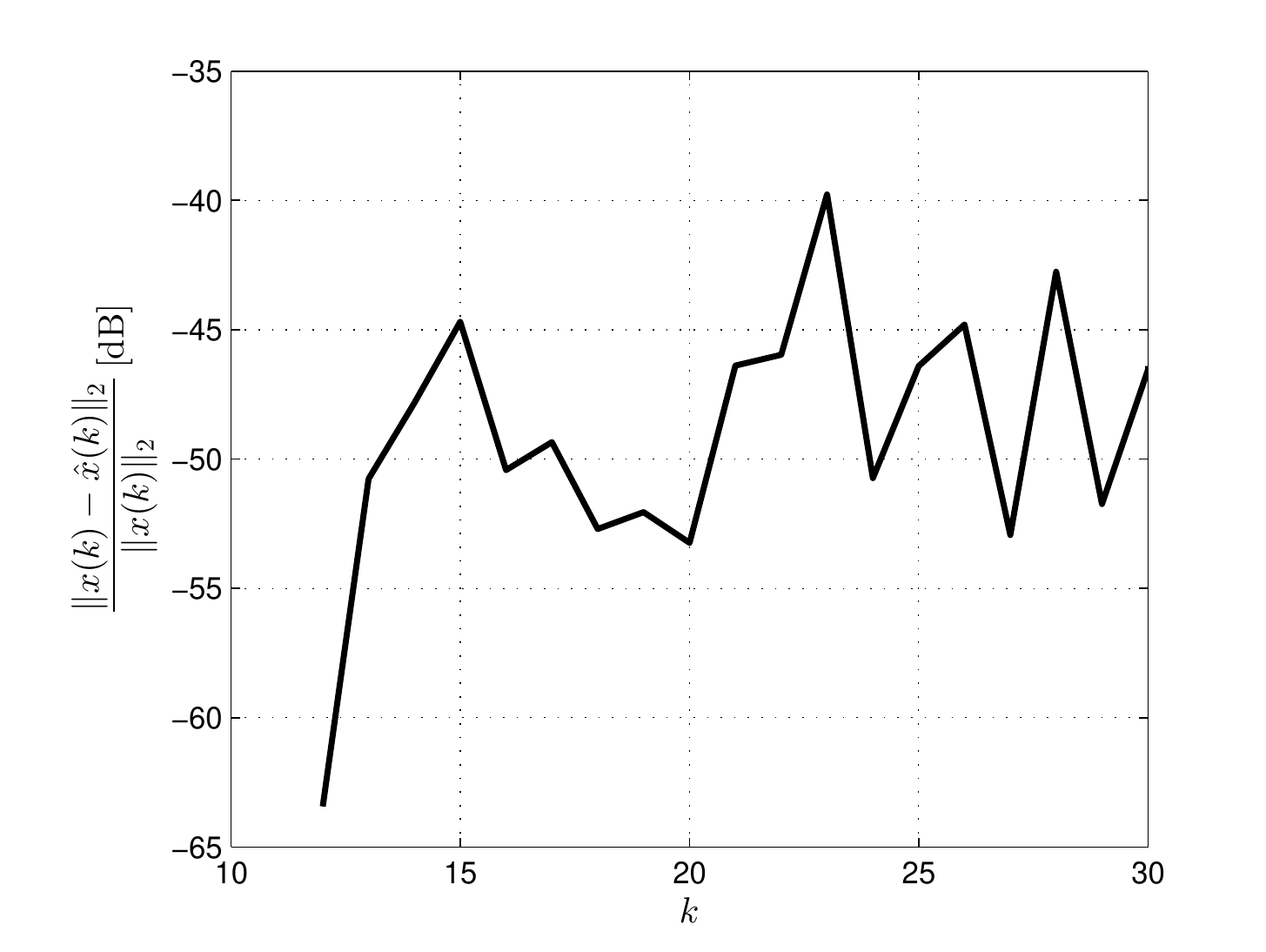}
\caption{\label{figureerror2} Scaled estimation error as a function of time when $p(k)=\lceil1.3s(k)\rceil$ for all $k\in\mathbb{N}_0$.}\vspace{-.09in}
\end{figure}

\begin{figure}[!t]
\centering
\includegraphics[width=0.45\linewidth]{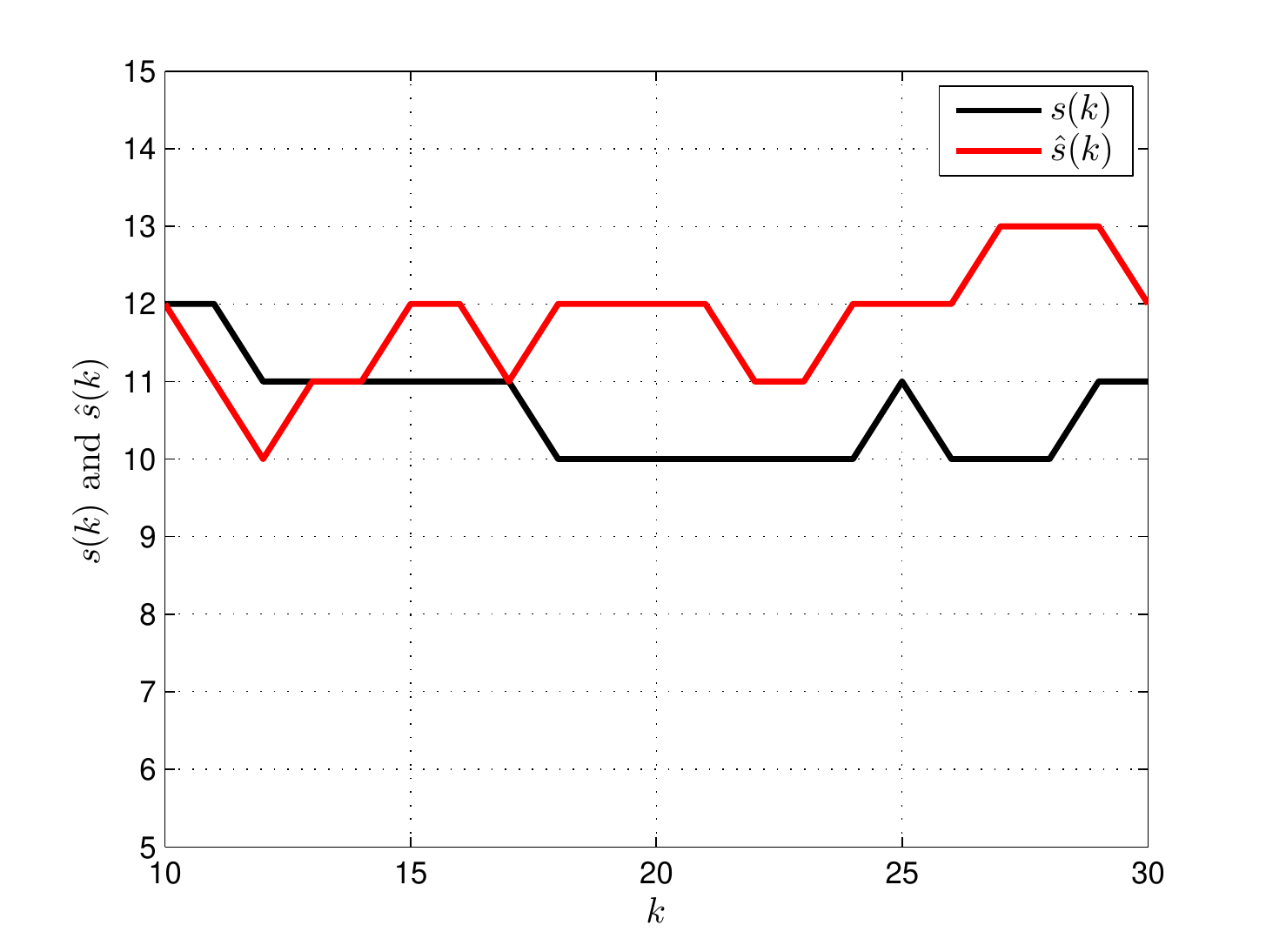}
\caption{\label{figureskhatsk6d} Sparsity factors $s(k)$ and $\hat{s}(k)$ as a function of time for the scenario described in Subsection~\ref{subsec:numerical:2}.}
\end{figure}

\section{Numerical Example} \label{sec:example}
Consider a discrete-time linear time-invariant dynamical system composed of $L=30$ subsystems, where each one can be described by
$$
x_\ell(k+1)=x_\ell(k)+\sum_{j=1}^L \alpha_{\ell j}(x_j(k)-x_\ell(k))+w_\ell(k),
$$
with initial conditions $x_\ell(0)$, $1\leq \ell\leq L$, that are chosen according to a normal probability distribution with zero mean and unit standard deviation. Let us rewrite this system in the form introduced in~(\ref{eqn:entiresystemmodel}). Doing so, we can define $A$ such that $a_{\ell j}=\alpha_{\ell j}$ if $\ell\neq j$ and $a_{\ell\ell}=1-\sum_{j\neq \ell}\alpha_{\ell j}$ otherwise. For the moment, we set $\tilde{A}=0$ (however, later, we introduce modeling uncertainties). Let us pick $\alpha_{\ell j}=5\times10^{-2}$. We illustrate the behavior of the encoding/decoding scheme in the following subsections. 
\vspace{-.08in}
\subsection{Perfect Modeling in the Presence of Two Active Exogenous Inputs} \label{subsec:numerical:1}\vspace{-.1in}
Let us pick the index set of active exogenous inputs as $\mathcal{L}=\{13,15\}$. Hence, we assume that $\{w_\ell(k)\}_{k=0}^\infty=0$ for any $\ell\notin\mathcal{L}$, but $\{w_\ell(k)\}_{k=0}^\infty$ is a stochastic process composed of independently and identically distributed random Gaussian variables with zero mean and unit variance for any $\ell\in\mathcal{L}$.

Now, we use the numerical procedure described by the block diagram in Figure~\ref{figure:1} to transmit the state measurements. Let us fix $m=30$, $H_b=10$, and $H_f=5$. For the first $H_b$ time steps, we let the transmitter send the whole state vector. This is to ensure that both bases (in the transmitter and the receiver) have access to a common history. Figure~\ref{figureskhatsk} shows $s(k)$ and $\hat{s}(k)$ versus time. Therefore, the transmitter and the receiver have learned bases that can represent the state vector of the system using a sparse vector. Figure~\ref{figureerror1} illustrates the error $\|x(k)-\hat{x}(k)\|_2$ and its upper bound in~(\ref{eqn:tho:0}) as function of time for the case where we transmit $p(k)=\lceil1.3s(k)\rceil$ measurements in each time step $k\in\mathbb{N}_0$ (which amounts to roughly half of the information, in terms of bits, that we need to relay in the case we do not compress the data). Figure~\ref{figureerror2} illustrates the scaled error $\|x(k)-\hat{x}(k)\|_2/\|x(k)\|_2$ (in $\mathrm{dB}$ scale) versus time. As we can see in Figure~\ref{figureerror2}, for this numerical example, the estimation error is practically negligible (with a noise-to-signal ratio of less than $-40\,\mathrm{dB}$ for all time steps).

\subsection{Imperfect Modeling in the Presence of Four Active Exogenous Inputs} \vspace{-.1in}
\label{subsec:numerical:2}
Let us change the index set of active exogenous inputs to $\mathcal{L}=\{2,4,26,28\}$.
Furthermore, assume that the modeling parameters $\alpha_{ij}$, $1\leq i,j\leq 30$, can deviate by $\pm2.5\%$ from the nominal model. Figure~\ref{figureskhatsk6d} shows $s(k)$ and $\hat{s}(k)$ versus time in this case. Comparing Figures~\ref{figureskhatsk} and~\ref{figureskhatsk6d}, we can deduce that by adding more uncertainty to the system, the estimator requires more measurements to recover the state vector of the system (with a noise-to-signal ratio of less than $-20\,\mathrm{dB}$ for all time steps).

\section{Conclusions and Future Research} \label{sec:conclusions}\vspace{-.1in}
In this paper, we proposed an encoding and decoding strategy for large-scale systems to transmit their entire state over a shared communication network. The strategy was based on finding a sparsifying basis that can replace the state vector by a vector with fewer nonzero components (i.e., a sparse vector). As a future direction of research, we can consider proposing an optimal estimator (based on Kalman filters for discrete-time linear time-varying system) to reduce the estimation error.

\bibliographystyle{ieeetr}
\bibliography{compile_new}

\end{document}